\documentclass[conference]{IEEEtran}
\IEEEoverridecommandlockouts
\usepackage{amsmath,amsfonts,amssymb,amscd,amsthm,xspace}
\usepackage{comment}
\usepackage[english]{babel}
\theoremstyle{plain}
\newtheorem{theorem}{Theorem}[section]

\newtheorem{lemma}[theorem]{Lemma}
\newtheorem{definition}[theorem]{Definition}

\usepackage{algorithmic}

\usepackage{graphicx}
\graphicspath{Figures/}

\usepackage{subcaption}

\usepackage[ruled,linesnumbered,lined,boxed,commentsnumbered]{algorithm2e}
\SetAlFnt{\scriptsize\sffamily}

\SetCommentSty{mycommfont}
\let\oldnl\nl
\newcommand{\nonl}{\renewcommand{\nl}{\let\nl\oldnl}}

\SetKwRepeat{Struct}{struct \{}{\}}
\usepackage{textcomp}
\usepackage{xcolor}
\def\BibTeX{{\rm B\kern-.05em{\sc i\kern-.025em b}\kern-.08em
    T\kern-.1667em\lower.7ex\hbox{E}\kern-.125emX}}

\usepackage{xtab}  
\usepackage{xcolor,colortbl}

\usepackage{multirow}   
\usepackage[para,online,flushleft]{threeparttable} 
\usepackage{verbatim}

\usepackage[style=ieee]{biblatex}

\begin{document}

\title{l2Match: Optimization Techniques on Subgraph Matching Algorithm using Label Pair, Neighboring Label Index, and Jump-Redo method\\}

\author{\IEEEauthorblockN{Chi Qin Cheng}
\IEEEauthorblockA{\textit{School of Information Technology} \\
\textit{Monash University}\\
Selangor, Malaysia \\
0000-0003-0718-3981}
\and
\IEEEauthorblockN{Kok Sheik Wong}
\IEEEauthorblockA{\textit{School of Information Technology} \\
\textit{Monash University}\\
Selangor, Malaysia \\
0000-0002-4893-2291}
\and
\IEEEauthorblockN{Lay Ki Soon}
\IEEEauthorblockA{\textit{School of Information Technology} \\
\textit{Monash University}\\
Selangor, Malaysia \\
0000-0002-8072-242X}
}

\maketitle

\begin{abstract}
Graph database is designed to store bidirectional relationships between objects and facilitate the traversal process to extract a subgraph. 
However, the subgraph matching process is an NP-Complete problem. 
Existing solutions to this problem usually employ a filter-and-verification framework and a divide-and-conquer method. 
The filter-and-verification framework minimizes the number of inputs to the verification stage by filtering and pruning invalid candidates as much as possible. 
Meanwhile, subgraph matching is performed on the substructure decomposed from the larger graph to yield partial embedding. 
Subsequently, the recursive traversal or set intersection technique combines the partial embedding into a complete subgraph. 
In this paper, we first present a comprehensive literature review of the state-of-the-art solutions. 
l2Match, a subgraph isomorphism algorithm for small queries utilizing a Label-Pair Index and filtering method, is then proposed and presented as a proof of concept. 
Empirical experimentation shows that l2Match outperforms related state-of-the-art solutions, and the proposed methods optimize the existing algorithms.
\end{abstract}


\section{Introduction}
Graph algorithm is an intensively researched subject. 
Its significance has grown in recent years. 
For instance, the shortest path algorithm is frequently used in the logistics and transportation sectors to cut waste and boost productivity~\cite{shortestlogistic}. 
Additionally, Natural Language Processing groups similar objects using the graph clustering technique to facilitate information retrieval~\cite{NLPgraph}. These are just a few, not all, instances of how the graph algorithm is commonly used. 
In a graph, edges depict the connections between vertices or objects. 
Protein-to-protein and social interactions are two examples of the kinds of information that can be represented with graphs. 
These data cannot be stored in a relational database due to the high level of randomness and uncertainty, as the schema and structure of the data must be identified and defined beforehand.

The Subgraph Matching problem, also known as Subgraph Isomorphism (\textbf{SI}), aims to match isomorphic subgraphs (embeddings) in a larger graph for a specific query graph. 
This procedure is somewhat comparable to matching the occurrences of a string pattern \(p\) in a text file \(X\). 
String matching is however a simpler problem than SI, where SI is NP-Complete and more challenging \cite{NPComplete}.

Recent researches \cite{CFL, CECI, DP-ISO, QuickSI, TurboIso} approach the SI problem with the filter-and-verification framework, which essentially trims the solution space to decrease the time consumption in the verification stage \cite{QuickSI}. 
In particular, CFL-Match\cite{CFL} and CECI\cite{CECI} apply Forward Candidate Generation and Backward Candidate Pruning (\textbf{FCGBCP}) to prune invalid candidates by leveraging the connection properties of a query vertex. 
Additionally, CECI proposes a Compact Embedding Cluster Index (\textbf{CECI}) auxiliary data structure to speed up index lookup during the filtering step. 
When moving forward through the query graph, Forward Candidate Generation (\textbf{FCG}) generates and filters the candidates. 
Backward Candidate Pruning (\textbf{BCP}), in contrast, filters and removes invalid candidates in the other direction. 
Nevertheless, the removal of an invalid candidate for a query vertex is not propagated to its neighboring candidates that are mapped to other query vertices. 
Therefore, the removal of invalid candidates is incomplete, necessitating another expensive refinement step while traversing backwards as observed in the CECI algorithm. 
Furthermore, various studies~\cite{CFL, CECI, DP-ISO} use the Neighboring Label Frequency (NLF) filtering method to compare the frequency of occurrence of each unique label among a vertex's neighbors between a query vertex and a candidate. 
It prunes a candidate if its NLF is lesser than the NLF of a query vertex for any neighboring unique label. 
Although the NLF is precomputed, the FCGBCP has to scrutinize every edge of a candidate to find its neighbors with a particular label. 
The above shortcomings motivate the proposal of l2Match algorithm in this research to optimise the filtering step of the CECI algorithm and reduce the total time consumption of query by: (1) traversing query graph efficiently during filtering step, and (2) avoiding unnecessary edge scanning in NLF filtering, and (3) reducing exploration of the redundant search branches in the enumeration step.

Our research makes following contributions:
\begin{enumerate}
    \item We develop the Label Pair Filtering (LPF) method that generates the set of candidate vertices for a given query vertex. LPF utilizes LPI to obtain a subset of data edges with a specific label pair. As a result, LPF skips the filtering process on the data vertices that are not connected to any neighbors of a specific label.
    \item We propose a Neighboring Label Index that extends the NLF index to omit the label and degree filtering on data edges with a mismatching label pair.
    \item We simplify FCGBCP filtering and traversal by combining the Backward Candidate Pruning (BCP) with the backward refinement step into the BCPRefine method to remove additional traversal on the query graph and avoid expensive invalid candidate removal operations and NLF filtering in the BCP step.
    \item We present the Jump-Redo (JR) technique to reduce the number of search nodes explored and the time taken in the enumeration step.
    \item We demonstrate the effectiveness of the l2Match algorithm that combines the LPI, LPF, Neighboring Label Index, CECI auxiliary indexing data structure, BCPRefine, and JR methods against the state-of-the-art algorithms through empirical evidence.
\end{enumerate}

\subsection{Preliminaries}

\begin{figure}
    \begin{subfigure}[t]{0.5\textwidth}
        \centering
        \includegraphics[width=0.2\textwidth]{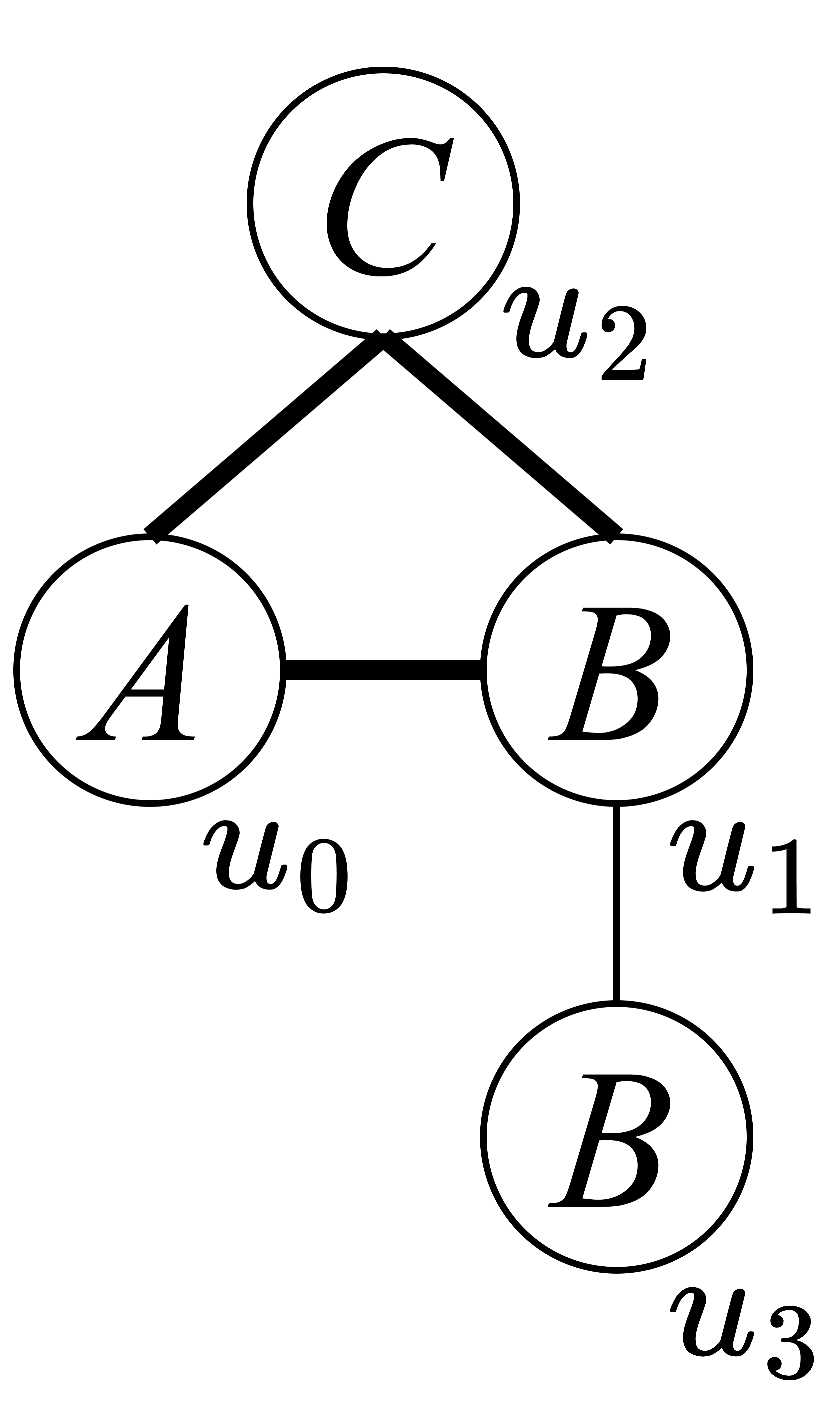}
        \caption{Query graph \(Q\).}
        \label{fig.query}
    \end{subfigure}\quad
    \begin{subfigure}[t]{0.5\textwidth}
        \centering
        \includegraphics[width=0.6\textwidth]{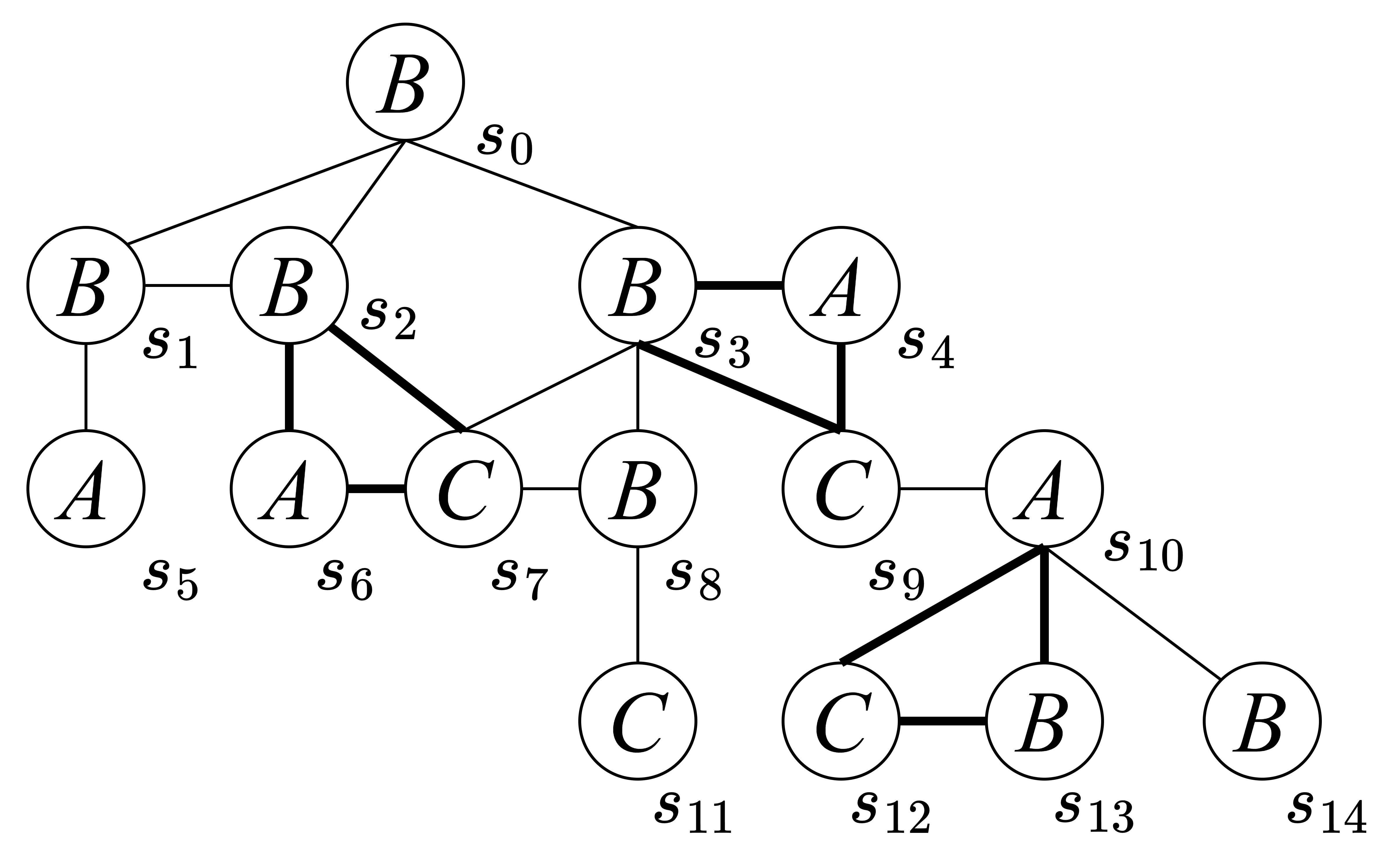}
        \caption{Data graph \(D\).}
        \label{fig.data}
    \end{subfigure}\
    \caption{Query Graph, Data Graph, and Embeddings.}
    \label{fig.QD}
\end{figure}

\topcaption{Notations.} \label{table.notation}
\tablefirsthead{\hline
    \textbf{Notation} & \textbf{Description} \\
\hline\hline}
\tablehead{
    \multicolumn{2}{c}%
    {{\captionsize\bfseries \tablename\ \thetable{} --
    continued from previous column}} \\
    \hline \textbf{Notation} & \textbf{Description} \\
\hline}
\tablelasthead{
    \multicolumn{2}{c}%
    {{\captionsize\bfseries \tablename\ \thetable{} --
    continued from previous column}} \\
    \hline \textbf{Notation} & \textbf{Description} \\
\hline}
\tabletail{\hline}
\tablelasttail{\hline \hline}

\begin{xtabular}{|p{0.4\linewidth}|p{0.5\linewidth}|}
        
        \(D\), \(Q\) & Data Graph and Query Graph \\ \hline
        
        \(V(D)\), \(E(D)\) & Data vertex set and data edge set \\ \hline
        
        \(V(Q)\), \(E(Q)\) & Query vertex set and query edge set \\ \hline
        
        \(e_D(s,t)\) & data edge connecting data vertex \(s\) and \(t\) \\ \hline
        
        \(e_Q(u,v)\) & query edge connecting query vertex \(u\) and \(v\) \\ \hline
        
        \(\Sigma\) & label set \\ \hline

        \(L:V\rightarrow\Sigma\) & Label function that maps each vertex in vertex set \(V\) to a label in label set \(\Sigma\) \\ \hline
        
        \(L(u)\) & Label of query vertex \(u\) \\ \hline

        \(N(u)\) & Neighbor vertex set of query vertex \(u\) \\ \hline

        \(N_l(u)\) & Neighbor vertex set of query vertex \(u\) with label \(l\) \\ \hline

        \(|N_l(u)|\) & Neighboring Label Frequency (\textbf{NLF}) of query vertex \(u\) with label \(l\) \\ \hline

        \(\omega_l(u)\) & Neighboring Label Offset (\textbf{NLO}) of query vertex \(u\) with label \(l\) \\ \hline

        \(\pi_l(u)=(\omega_l(u), \, |N_l(u)|)\) & Pair of NLO and NLF of query vertex \(u\) with label \(l\) \\ \hline

        \(|u|\) & Degree of query vertex \(u\) \\ \hline

        \(C(u)=\{s\in V(D)\, |\, |s|\geqslant|u|,\, L(s)=L(u)\}\) & Candidate set of data vertices that have similar label and greater or equal degree to the query vertex \(u\) \\ \hline
        
        \(I\) & Index data structure \\ \hline
        
        \(I_{v}^{u}(s)\) & Neighbor vertex set \(N(s)\) of data vertex \(s\) in \(C(v)\) where \(s\in C(u)\) and \(e_Q(u,v)\in E(Q)\) \\ \hline
        
        \(\Phi\) & Label Pair Index (\textbf{LPI}) \\ \hline
        
        \(\Phi_{l_1}\) & Data vertex set with label \(l_1\) \\ \hline
        
        \(\Phi_{l_2}^{l_1}\) & Data edges that have one data vertex with label \(l_1\) and a neighbor data vertex with label \(l_2\) \\ \hline
        
        \(T\) & Breadth First Search (\textbf{BFS}) tree decomposition of query graph \(Q\) \\ \hline

        \(u.p, u.c\) & Parent and child query vertex of query vertex \(u\) in BFS tree \(T\) \\ \hline
        
        \(o\) & Enumeration order is a traversal sequence of query vertex in query graph \(Q\) during enumeration step \\ \hline
        
        \(N_{-}^{o}(u)\) & Neighbor query vertex set of query vertex \(u\) that are positioned ahead of \(u\) in the enumeration order \(o\) \\ \hline

        \(N_{+}^{o}(u)\) & Neighbor query vertex set of query vertex \(u\) that are positioned after \(u\) in the enumeration order \(o\) \\ \hline

        \(\mu\) & Partial mapping of incomplete subgraph of data graph \(D\) to query graph \(Q\) \\ \hline

        \(LC(u)= \begin{cases}
            \begin{aligned}[c]
                \{s\in C(u)\,| \\
                \exists e_D(s,\mu[u.b]) \\
                \in E(D) \\
                \forall u.b\in N_{-}^{o}(u)\} \\
                if\, |N_{-}^{o}(u)|>0
            \end{aligned} \\
            \begin{aligned}[c]
                C(u) \\ 
                if\, |N_{-}^{o}(u)|=0
            \end{aligned}
        \end{cases} \) & If query vertex \(u\) is not root vertex in \(T\), the local candidate set is a subset of \(C(u)\) for the query vertex \(u\) such that each candidate \(s\in LC(u)\) is connected to all candidate of backward neighbors of \(u\) in \(\mu\), otherwise \(LC(u)=C(u)\) \\ \hline

        \(\mu[u]\) & A data vertex \(s\in LC(u)\) that is mapped to \(u\) in \(\mu\) \\ \hline
        
        \(f:V(Q)\rightarrow V(D)\) & An injective function that maps query vertices to data vertices such that for all query edge \(e_Q(u,v)\in E(Q)\), \(L(u)=L(f(u))\), \(L(v)=L(f(v))\), \(|f(u)|\geqslant|u|\), \(|f(v)|\geqslant|v|\), and \(e_D(f(u),f(v))\in E(D)\). \\
    \hline
\end{xtabular}

Given a query graph \(Q\), a data graph \(D\), the set of query vertices \(V(Q)\) and data vertices \(V(D)\), the set of query edges \(E(Q)\) and data edges \(E(D)\), and a label function \(L\) that maps each vertex in \(V(Q)\) and \(V(D)\) to a set of labels \(\Sigma\), non-induced SI is defined as an injective function \(f:V(Q)\rightarrow V(D)\) that maps the query vertices \(V(Q)\) to data vertices \(V(D)\) such that for all query edge \(e_Q(u,v)\in E(Q)\), \(L(u)=L(f(u))\), \(L(v)=L(f(v))\), \(|f(u)|\geqslant|u|\), \(|f(v)|\geqslant|v|\), and \(e_D(f(u),f(v))\in E(D)\). 
For instance, the set of vertices of an isomorphic subgraph that can be mapped to \(\{u_3, u_1, u_0, u_2\}\) in the query graph \(Q\) are shown in Figure \ref{fig.QD} as \(\{\)\(\{s_0, s_2, s_6, s_7\}\), \(\{s_1, s_2, s_6, s_7\}\), \(\{s_0, s_3, s_4, s_9\}\), \(\{s_8, s_3, s_4, s_9\}\)\(\}\). 
Each vertex's label is represented by an alphabet within it. Table \ref{table.notation} outlines the frequently-used notations in this paper.

\section{Related Work}
There are varying methods for approaching the SI problem, such as query graph decomposition, search space reduction, and query parallelization.

\subsection{Indexing Method}
SubGlw decomposes the data graph into small candidate graphs using a ranking function based on candidate count and maximum distance between data vertices \cite{subglw}. 
Nonetheless, the decomposition of the data graph takes an additional \(O(|V'|)\) time. 
VC \cite{VC} implements a bipartite graph index, but the cost (size) of each query vertex is calculated in advance before filtering and pruning procedures. 
Meanwhile, SMS2 \cite{sms2} employs lattice-based index and hashing algorithm as an alternative SI indexing approach. 
QuickSI \cite{QuickSI} sorts the ordering of edges in the query graph based on the frequency of the edges presented in the data graph. 
TurboISO \cite{TurboIso} and CFL-Match \cite{CFL} implement a tree-based index to reduce pruning cost and iteration. Furthermore, the matching is performed at the index level instead of the data level. 
However, both TurboISO \cite{TurboIso} and CFL-Match \cite{CFL} have the common drawback that the ordering of edges is not considered when traversing the index.

\subsection{Reducing Search Space}
In terms of reducing search space, VEQ\textsubscript{M} employs dynamic equivalence to avoid enumerating candidates that would not yield any embedding if a candidate that shares similar neighbors rooted at a subtree in Candidate Space (CS) fails to yield any embedding~\cite{veqm}. 
Here, embedding refers to a complete subgraph that is isomorphic to the query graph \(Q\). 
Besides, HyGraph utilized branching in a potential isomorphism boundary to select candidates \cite{hygraph}. The branching approach in HyGraph is prone to automorphism, i.e., the permutation of partial embedding. 
VF2 suggests the employment of ordered relationships and State Space Representation (SSR), which resembles an automaton machine where certain conditions must be fulfilled to proceed to the next state, also known as the feasibility rules. Candidate graphs are collected and refined progressively while transiting through states in SSR. 
Furthermore, GraphQL \cite{GQL} proposes three steps to reduce search space, such as neighbourhood subgraph string profile, joint reduction, and search ordering using a cost model. Nonetheless, GraphQL \cite{GQL} does not support recursive pattern search.

\subsection{Parallel Processing}
Parallel processing technique exploits the advantages of processing data simultaneously to increase the efficiency of an algorithm. 
Glasgow is a constraint programming SI algorithm adopting the parallel processing method \cite{glasgow}. It employs a bit-parallel data structure and threads to explore search spaces in parallel. 
Similarly, G-Morph algorithm makes use of the advantage of parallel processing on a GPU to solve induced SI for labeled graphs \cite{gmorph}. 
Last but not least, GR1 is a parallel processing SI algorithm that presents a transition from parallelism to a quantum computing approach to solve the SI problem \cite{gr1}. It facilitates a producer-consumer pattern that allows interchanges of varying existing pruning strategies to adapt to different use cases, though it has only experimented on the protein-to-protein interaction graph.

\section{l2Match}

\subsection{Overview}
\begin{algorithm}
\SetAlgoLined
\SetKwProg{Pn}{Function}{:}{end}
\SetKwProg{Fn}{Object}{:}{end}
\SetKwFunction{constructLPI}{constructLPI}
\SetKwFunction{convertToCSR}{convertToCSR}
\SetKwFunction{Filter}{Filter}
\SetKwFunction{Ordering}{Ordering}
\SetKwFunction{Enumerate}{Enumerate}
\SetKwInOut{Input}{Input}\SetKwInOut{Output}{Output}
\Input{Query graph $Q$, Data graph $D$}
\Output{All matched subgraph isomorphic to $Q$}

$Q\leftarrow$ \convertToCSR{$Q$} \tcp*{Convert $Q$ to CSR}
$D,\Phi,\pi\leftarrow$ \constructLPI{$D$} \tcp*{Construct Label Pair Index}
$C,I,o\leftarrow$ \Filter{$Q,D,\Phi,\pi$} \tcp*{Ordering, filtering and index construction}
\Enumerate{$Q,D,C,I,o,\{\},0$} \tcp*{Recursive enumeration}

\caption{l2Match.}
\label{algo:main}
\end{algorithm}

l2Match takes the form of filter-and-verification framework described in \cite{CFL, PLGCoding}. 
Algorithm \ref{algo:main} shows the highly-abstracted pseudocode of l2Match. 
It firstly constructs a Label Pair Index (\textbf{LPI}) on the query and data graph in Section \ref{subsec:LPI}. Subsequently, the filtering step prunes the invalid candidates and stores the partial candidate-to-query vertex mappings in an auxiliary indexing data structure, Compact-Embedding Cluster Index (\textbf{CECI}) identified with the notation \(I\). The filtering step comprises of a Label Pair Filtering (\textbf{LPF}) technique defined in Section \ref{subsec:lpf} to filter and map the candidates to each query vertex with label pair constraint, Breadth-First Search Ordering method to sort the traversal order of the query vertices in the enumeration step, \textbf{FCG} to generates and filters the candidates for each query vertex, and \textbf{BCPRefine} defined in Section \ref{subsec:bcprefine} to prune the invalid candidates and update the auxiliary indexing data structure. 
Finally, the enumeration step recursively verifies and extends the partial mapping in the auxiliary indexing data structure into an embedding. l2Match uses Jump-Redo (\textbf{JR}) method defined in Section \ref{subsec:JR} to reduce the exploration of redundant search branches in the enumeration step. The novelty of this research is observed in LPI, LPF, BCPRefine, and JR method. 

\subsection{Neighboring Label Index}
\begin{figure}[!ht]
\centering
  \includegraphics[width=0.2\textwidth]{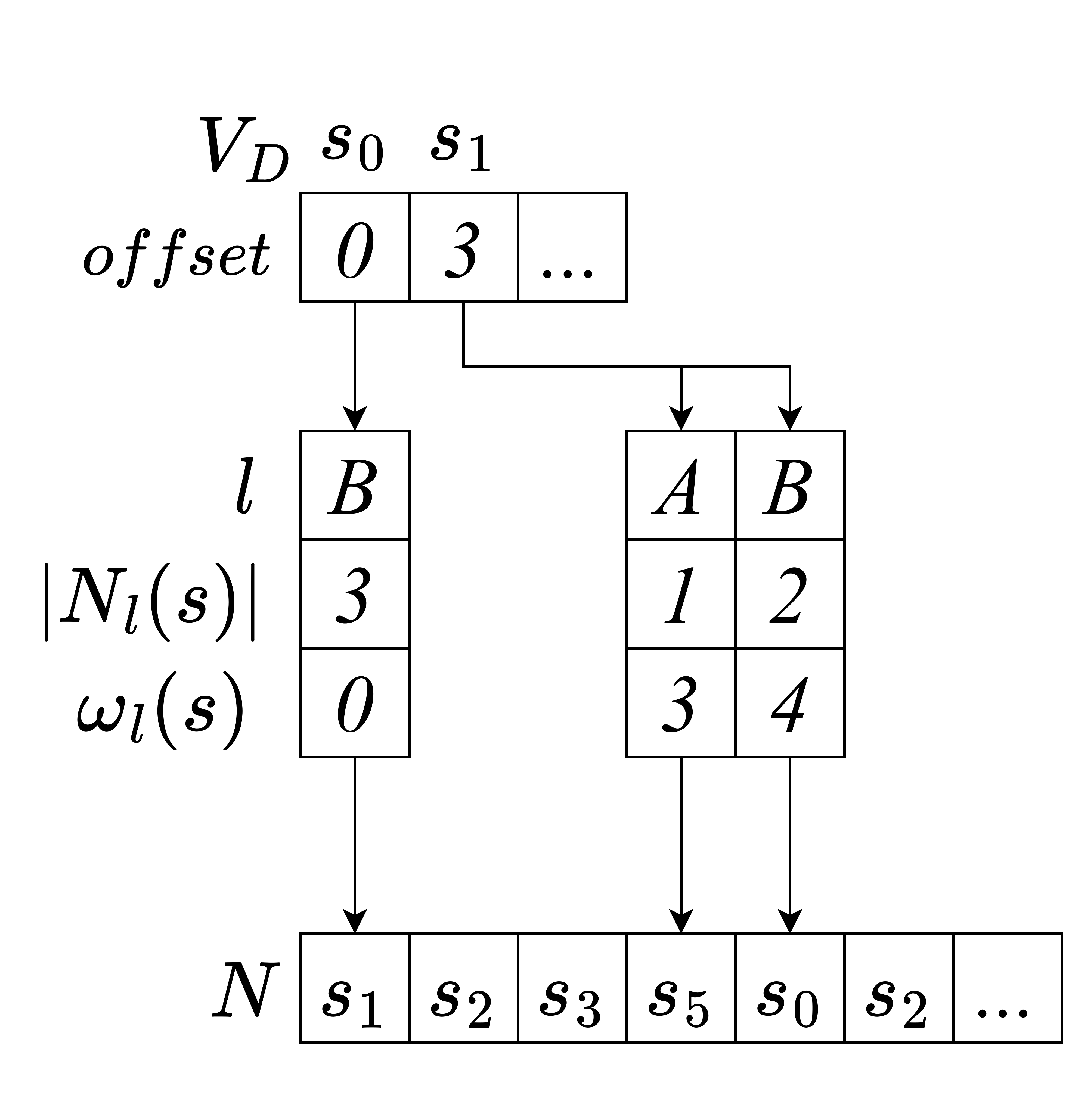}%
  \caption{Neighboring Label Index \(\pi\)}
  \label{fig.NLO}
\end{figure}

Neighboring vertices of any vertex \(u\) may have identical or dissimilar labels. The size of each group with a \emph{distinct} label \(\forall l\in L(N(u))\) is defined as Neighboring Label Frequency (NLF) in \cite{CFL, VC, treespan}, or formally \(|N_l(u)|\). NLF is computed and generated during preprocessing operations on the input graphs. 
Neighboring Label Offset (NLO), extends the concept of NLF by enabling direct access to a vertex's neighbors with a particular label. 
FCG and BCPRefine filtering technique employ NLO to skip validation on the neighbors of any candidate with an undesired label. 
Given a label \(l\), the NLO \(\omega_l(u)\) of a vertex \(u\) stores the offset position to the neighbor vertex set \(N_l(u)\) in the sorted adjacent array \(N\) of a Compressed Sparse Row (\textbf{CSR}). The combination of NLO and NLF of query vertex \(u\) with label \(l\) is defined as Neighboring Label Index \(\pi_l(u)=\bigl( \omega_l(u),\, |N_l(u)| \bigr)\) and computed during the computation of NLF.

\subsection{Label Pair Index}
\label{subsec:LPI}
\begin{figure}[!ht]
\centering
  \includegraphics[width=0.4\textwidth]{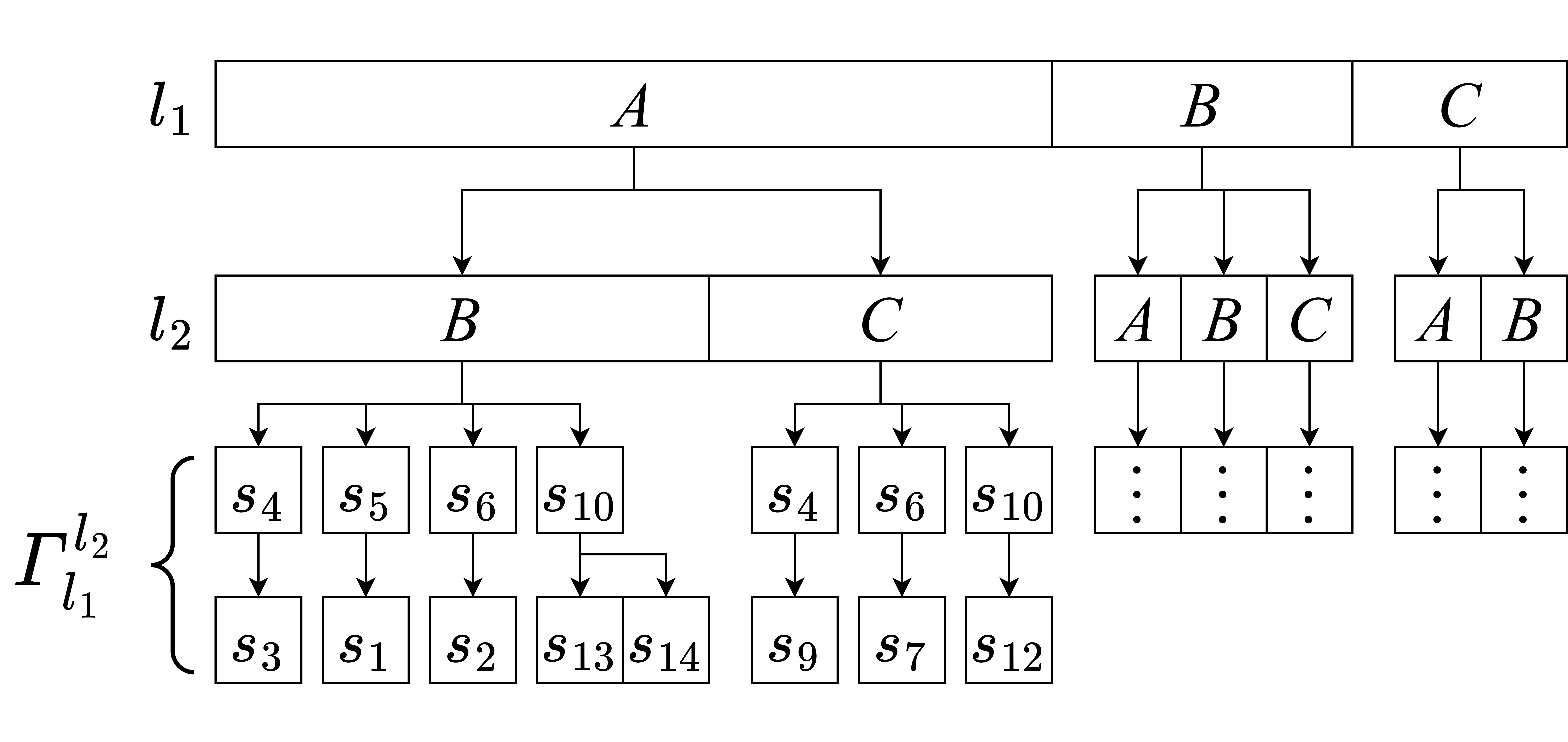}%
  \caption{Label Pair Index \(\Phi\)}
  \label{fig.LPI}
\end{figure}

Each edge in the context of the data graph \(D\) connects two vertices, for example, \(e_D(s,t)\in E(D)\). 
As a result, the labels of the two ends of an edge form a label pair \(\bigl(L(s),L(t)\bigr)\). Label Pair Index (\textbf{LPI}), denoted by \(\Phi\), groups data edges \(E(D)\) by the label pair of each edge into an index data structure. 
It allows direct access to data edges with the desired label pair \(\bigl(L(s),L(t)\bigr) = \bigl(L(u),L(v)\bigr)\). This approach reduces the number of iterations down to \(\Sigma_{i=0}^{|u|} |\Phi_{L(u)}^{L(v_i)}| \quad \forall v_i\in N(u)\).

\subsection{Label Pair Filtering}
\label{subsec:lpf}
LPF employs LPI to yield a minimal candidate set \(C(u)\) for any query vertex \(u\). It computes the smallest subset of data edges with the label \(L(u)\) and neighboring labels \(L\bigl(N(u)\bigr)\). The neighboring labels of \(u\) that generate the minimum subset of data edges are referred to as \(l_{min}\). Following this, it performs degree and NLF filtering. As a result, LPF skips the degree and NLF filtering on the data vertices that have a different label than \(u\), or are not connected to any neighbors with the label \(l_{min}\). LPF uses the edge filtering technique for greater pruning power compared to conventional methods, which use the vertex filtering technique to obtain a preliminary set of data vertices before degree and NLF filtering. 
In other words, the smaller the preliminary set of data vertices to process, the less filtration needs to be performed.

\subsection{Backward Candidate Pruning and Refinement}
\label{subsec:bcprefine}
\begin{lemma}
\label{lemma:removefromparent}
Removing just an invalid candidate \(s\) from parent index \(I_{u}^{u.p}(t)\) and \(C(u)\) 
is sufficient to prevent \(s\) from being returned as false positive candidate of \(LC(u)\) in the enumeration step.
\end{lemma}
\begin{proof}
\label{theorem:removefromparent}
Assume that removing just an invalid candidate \(s\) from \(I_{u}^{u.p}(t)\) and \(C(u)\) where \(s\in C(u)\) and \(e_D(s,t)\in E(D) \quad \forall t\in C(u.p)\) will cause \(s\) to be perceived as valid candidate in local candidate set \(LC(u)=\bigcap_{u.b\in N_{-}^{o}(u)} I_{u}^{u.b}(\mu[u.b])\) of \(u\) during the enumeration step. 
Since \(s\notin I_{u}^{u.p}(t) \quad \forall t\in C(u.p)\) and \(u.p\in N_{-}^{o}(u)\), thus \(s\notin LC(u)\). The lemma holds by contradiction.
\end{proof}

\begin{lemma}
\label{lemma:notremovefromchild}
Retaining an invalid candidate \(s\) from index of forward neighbor \(I_{u.f}^{u}(s) \quad \forall u.f\in N_{+}^{o}(u)\) will not yield \(s\) as a false positive candidate mapped to \(u\) in the enumeration step.
\end{lemma}
\begin{proof}
\label{theorem:notremovefromchild}
Assume that retaining an invalid candidate \(s\) in \(I_{u.f}^{u}(s) \quad \forall u.f\in N_{+}^{o}(u)\) where \(e_D(s,t)\in E(D) \quad \forall t\in I_{u.f}^{u}(s)\) will cause \(s\) to be taken as valid candidate for computing local candidate set \(LC(u.f)=\bigcap_{u.f\in N_{+}^{o}(u)} I_{u.f}^{u}(\mu[u])\). Since \(s\notin I_{u}^{u.p}(t) \quad \forall t\in C(u.p)\) and \(s\notin LC(u)\) according to Lemma \ref{lemma:removefromparent}, \(s\) is never mapped to \(u\) as \(\mu[u]\). This contradicts the assumption. Hence, we proof by contradiction that the lemma holds.
\end{proof}

\begin{definition}
\label{def:BCPRefine}
\textit{Backward Candidate Pruning Refinement (BCPRefine) Rule:} A candidate \(s\in C(u)\) of a query vertex \(u\) that is not connected to any candidate \(t\in C(u.b)\) of one or more backward neighbors \(u.b\in N_{-}^{o}(u)\) can be safely pruned, such that $\exists u.b\in N_{-}^{o}(u), N(s)\cap C(u.b) = \varnothing$. Since \(e_D(s,t)\notin E(D) \quad \forall t\in C(u.b)\) for some \(u.b\in N_{-}^{o}(u)\) but \(e_Q(u,v)\in E(Q)\), thus \(s\) violates the injective function \(f\).
\end{definition}

BCPRefine integrates the backward refinement step into the Backward Candidate Pruning (\textbf{BCP}) method through the utilization of two observations, namely Lemma \ref{lemma:removefromparent} and \ref{lemma:notremovefromchild}, thus reducing the expensive removal of invalid candidates from the candidate set \(C\) and auxiliary index \(I\) despite consuming more memory. The combined method, BCPRefine, employs a filtering constraint as shown in Definition \ref{def:BCPRefine}. It performs the filtering constraint on any query vertex and its backward neighbors, excluding the root vertex in the BFS tree \(T\). As a result, BCPRefine is less complex than the ordinary BCP method but has weaker pruning power.

\subsection{Jump and Redo}
\label{subsec:JR}
Jump and Redo (\textbf{JR}) is an optimization method that reduces search branches in the enumeration step given a static enumeration order. 
The result is a reduction in the quantity of data vertices to be explored. 
If a search branch on the current query vertex fails, the enumeration step jumps to the nearest backward neighbor in the enumeration order. 
Meanwhile, each query vertex between the jump (the source and the target query vertices) is mapped to the first candidate in its respective local candidate set. 
Recalculating the local candidates set for each intermediate query vertex will yield the same local candidates, inclusive of candidates that has been mapped to other query vertices, unless the mapping of the vertex's backward neighbors has been changed. This is because the computation of the local candidates of any query vertex is dependent on the mapping of its backward neighbors. As a result, the enumeration step will only recompute the local candidate set for any query vertex if the mapping of its backward neighbors has altered.

\section{Environment Setup}
\label{sec:environment}
All code is compiled in C++20 programming language with cmake (version 3.20.1, linux-x86\_64 distribution) g++ (GCC) v10.2.0 using O3 optimization and additional flags ("-march=native", "-pthread") on a machine with Intel Xeon W-2145 3.70GHz base clock speed, eight cores processor and 64GB RAM running Ubuntu 18.04.5 LTS OS. 

\section{Implemented Algorithm}
\label{sec:implementedAlgo}
Four state-of-the-art SI indexing algorithms, DP-iso (DPiso)\cite{DP-ISO}, GraphQL (GQL)\cite{GQL}, CFL-Match (CFL)\cite{CFL} and CECI\cite{CECI} are implemented to empirically compare against the proposed method, l2Match. We repurposed the source code of the state-of-the-art algorithms implemented by the researchers from a survey \cite{shixuan-survey}. 
Their codes can be found in the public repository \url{https://github.com/RapidsAtHKUST/SubgraphMatching/}. 
We also integrate the LPF filtering method and JR optimization technique into the CECI, CFL, GQL, and DPiso algorithms, namely CECI-LPF, CFL-LPF, GQL-LPF, DPiso-LPF, CECI-JR, CFL-JR, and GQL-JR algorithms, to measure the effectiveness of the optimization method. 
DP-iso is not compatible with the JR optimization method, as JR only works on static enumeration order but not dynamic enumeration order.

\section{Inputs}
\label{sec:inputs}

This research considers five real-world datasets (DBLP, Human, Patents, Yeast, and Youtube) and a publicly available synthetic dataset (ERP), undirected in this context and widely used in previous work \cite{exp3,QuickSI,exp2,exp1,hardness}, to evaluate the performance of the SI algorithms. The specifications of the datasets are described in Table \ref{table:dataset-query} and Table \ref{table:dataset-data}. 
Human and Yeast use edges to represent the interaction between proteins labeled using the Gene Ontology Term. 
Labels are generated at random and assigned to unlabeled datasets such as Patents, Youtube and DBLP. 
ERP consists of randomly generated Erdős–Rényi graphs with query graphs that are close to the phase transition between the satisfiable-unsatisfiable regions defined in \cite{satisfiable}. 
Satisfiability denotes whether a query graph is isomorphic to at least one subgraph in a given data graph, such that the query graph is satisfiable if true and otherwise unsatisfiable. 
On the other hand, the query graphs for the real-world datasets are induced subgraphs extracted from each dataset using the random walk method. 
These query instances are guaranteed to be satisfiable. 
The density of a graph is calculated using the formula \(density=\frac{2|E|}{\left(|V|(|V|-1)\right)}\).

\begin{table}[t]
    \caption{Properties of Query Datasets.\tnote{1}}
    \label{table:dataset-query}
    \centering
    \resizebox{0.5\textwidth}{!}{
        \begin{threeparttable}
            \begin{tabular}{|l|l r|l r|l r|l r|}
                \hline
                    \multirow{3}{*}{Category} & \multirow{3}{*}{Dataset} & \multirow{3}{*}{\#inst} & \multicolumn{6}{c|}{Query Graph} \\ \cline{4-9}
                    & & & \multicolumn{2}{c|}{$|V(Q)|$} & \multicolumn{2}{c|}{$|E(Q)|$} & \multicolumn{2}{c|}{density} \\ \cline{4-9}
                    & & & min & max & min & max & min & max \\
                \hline \hline
                    \multirow{2}{*}{\textbf{Biology}} & \textbf{Human} & 900 & 4 & 20 & 3 & 190 & .15 & 1 \\ \cline{2-9}
                    & \textbf{Yeast} & 900 & 4 & 32 & 3 & 88 & .09 & 1 \\ \hline \hline
                    
                    \textbf{Citation} & \textbf{Patents} & 900 & 4 & 32 & 3 & 210 & .09 & 1 \\ \hline \hline
                    
                    \multirow{2}{*}{\textbf{Social}} & \textbf{Youtube} & 900 & 4 & 32 & 3 & 89 & .09 & .83 \\ \cline{2-9}
                    & \textbf{DBLP} & 900 & 4 & 32 & 3 & 335 & .09 & 1 \\ \hline \hline
        
                    \textbf{Synthetic} & \textbf{ERP} & 200 & 30 & 30 & 128 & 387 & .29 & .89 \\
                \hline
            \end{tabular}
            \begin{tablenotes}
                \item[1] This table reports the specifications of each dataset with number of query graphs (\#inst), minimum and maximum number of query vertices ($|V(Q)|$), query edges ($|E(Q)|$) and density.
            \end{tablenotes}
        \end{threeparttable}
    }
\end{table}

\begin{table}[t]
    \caption{Properties of Data Graph.\tnote{1}}
    \label{table:dataset-data}
    \centering
    \resizebox{0.5\textwidth}{!}{
        \begin{threeparttable}
            \begin{tabular}{|l|l|l r|l r|l|r|l r|}
                \hline
                    \multirow{3}{*}{Category} & \multirow{3}{*}{Dataset} & \multicolumn{8}{c|}{Data Graph} \\ \cline{3-10}
                    & & \multicolumn{2}{c|}{\multirow{2}{*}{$|V(D)|$}} & \multicolumn{2}{c|}{\multirow{2}*{$|E(D)|$}} & \multirow{2}{*}{$|\Sigma|$} & average & \multicolumn{2}{c|}{\multirow{2}{*}{density}} \\
                    & & & & & & & degree & & \\
                \hline \hline
                    \multirow{2}{*}{\textbf{Biology}} & \textbf{Human} & 
                    \multicolumn{2}{l|}{4,674} & 
                    \multicolumn{2}{r|}{86,282} & 
                    44 & 
                    36.9200 & 
                    \multicolumn{2}{c|}{0.007901} \\ \cline{2-10}
                    & \textbf{Yeast} & 
                    \multicolumn{2}{l|}{3,112} & 
                    \multicolumn{2}{r|}{12,519} & 
                    71 & 
                    8.0456 & 
                    \multicolumn{2}{c|}{0.002586} \\ \hline \hline
                    
                    \textbf{Citation} & \textbf{Patents} & 
                    \multicolumn{2}{l|}{3,774,768} & 
                    \multicolumn{2}{r|}{16,518,947} & 
                    20 & 
                    8.7523 & 
                    \multicolumn{2}{c|}{0.000002} \\ \hline \hline
                    
                    \multirow{2}{*}{\textbf{Social}} & \textbf{Youtube} & 
                    \multicolumn{2}{l|}{1,134,890} & 
                    \multicolumn{2}{r|}{2,987,624} & 
                    25 & 
                    5.2650 & 
                    \multicolumn{2}{c|}{0.000005} \\ \cline{2-10}
                    & \textbf{DBLP} & 
                    \multicolumn{2}{l|}{317,080} & 
                    \multicolumn{2}{r|}{1,049,866} & 
                    15 & 
                    6.6221 & 
                    \multicolumn{2}{c|}{0.000021} \\ \hline
        
                    \multicolumn{2}{c}{ } & \multicolumn{1}{l}{min} & \multicolumn{1}{r}{max} & \multicolumn{1}{l}{min} & \multicolumn{1}{r}{max} & \multicolumn{2}{r}{ } & \multicolumn{1}{l}{min} & \multicolumn{1}{r}{max} \\ \hline
                    
                    \textbf{Synthetic} & \textbf{ERP} & 
                    150 & 150 & 
                    4132 & 8740 & 
                    0 & 
                    91.2718 & 
                    .37 & .78 \\
                \hline
            \end{tabular}
            \begin{tablenotes}
                \item[1] This table reports the specifications of each dataset with the number of data vertices ($|V(D)|$), data edges ($|E(D)|$), number of labels ($|\Sigma|$), average degree and density.
            \end{tablenotes}
        \end{threeparttable}
    }
\end{table}

\section{Metrics}
\label{sec:metrics}
The time taken for the filtering step (filtering time), the enumeration step (enumeration time), and the total time taken (query time) are measured in nanoseconds, however reported in seconds (s). This research excludes memory consumption as a metric, as all competing algorithms never exceed the maximum memory capacity. Meanwhile, the total candidate count is calculated as \(\Sigma_{u\in V(Q)}|C(u)|\). The enumeration step of each algorithm is halted at \(10^{6}\) embeddings to allow sufficient coverage of search space within a reasonable amount of time in accordance to \cite{CFL,shixuan-survey}. In addition, we set the time limit to \(0.3\times 10^{3}\) seconds (5 minutes) for the enumeration step. Any query that cannot be solved by one or more algorithms within the time limit is excluded from the computation of the average query time. The number of query edges are mostly greater than the number of query vertices except for some query graphs with \(|V(Q)|=4\) and \(|E(Q)|=3\). Hence, we report the performance evaluation over the number of query edges.

\section{Results and Analysis}
\label{sec:results}

\subsection{Hardness Analysis}

\begin{figure}
    \begin{subfigure}[t]{\columnwidth}
        \centering
        \includegraphics[width=\linewidth]{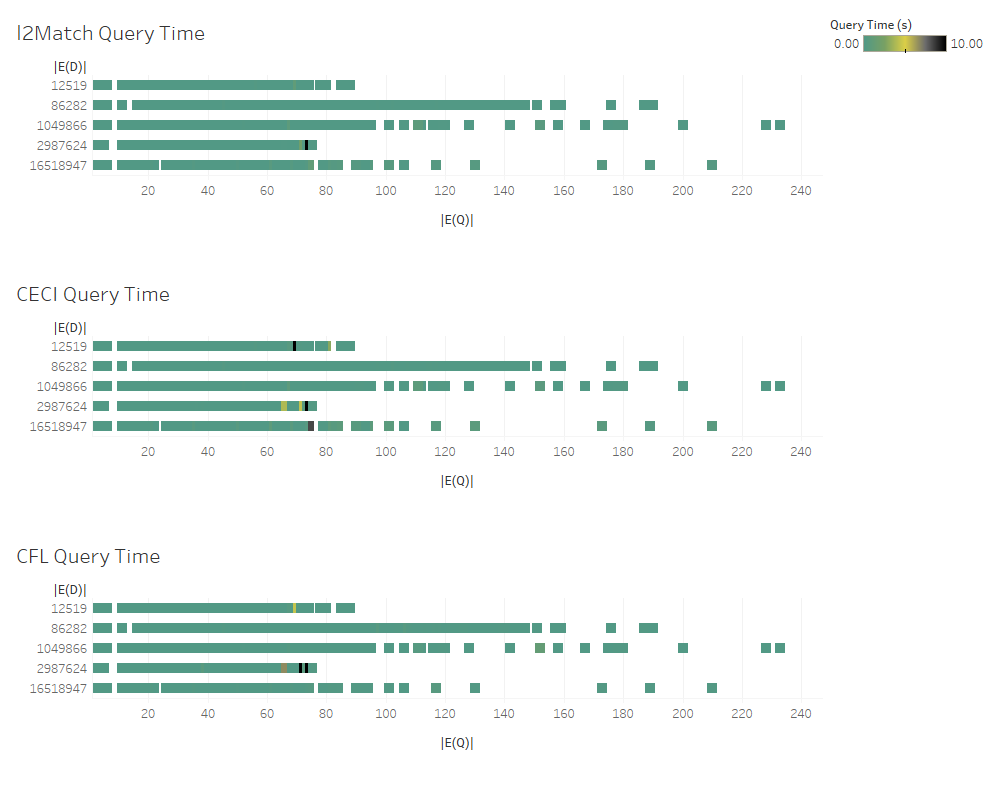}
    \end{subfigure}\quad
    \begin{subfigure}[t]{\columnwidth}
        \centering
        \includegraphics[width=\linewidth]{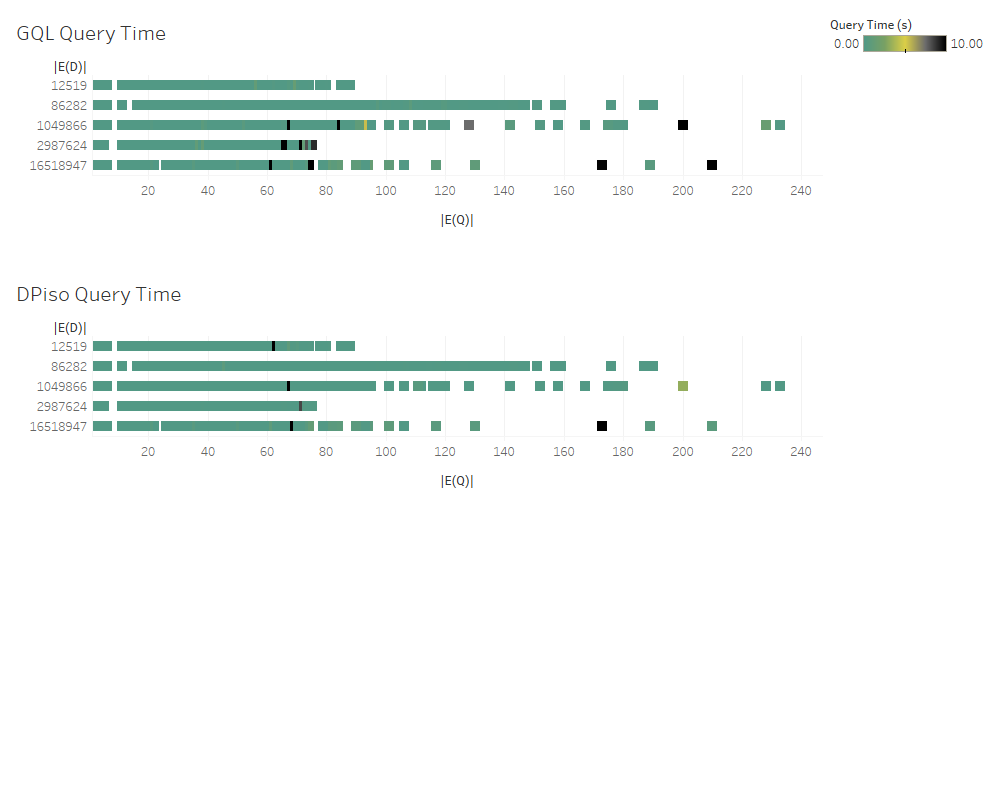}
    \end{subfigure}
    \caption{Number of query edge (x-axis), number of data edges (y-axis), and query time (color).}
    \label{fig.scatter_query}
\end{figure}

\begin{figure}
    \begin{subfigure}[t]{0.5\textwidth}
        \centering
        \includegraphics[width=\textwidth]{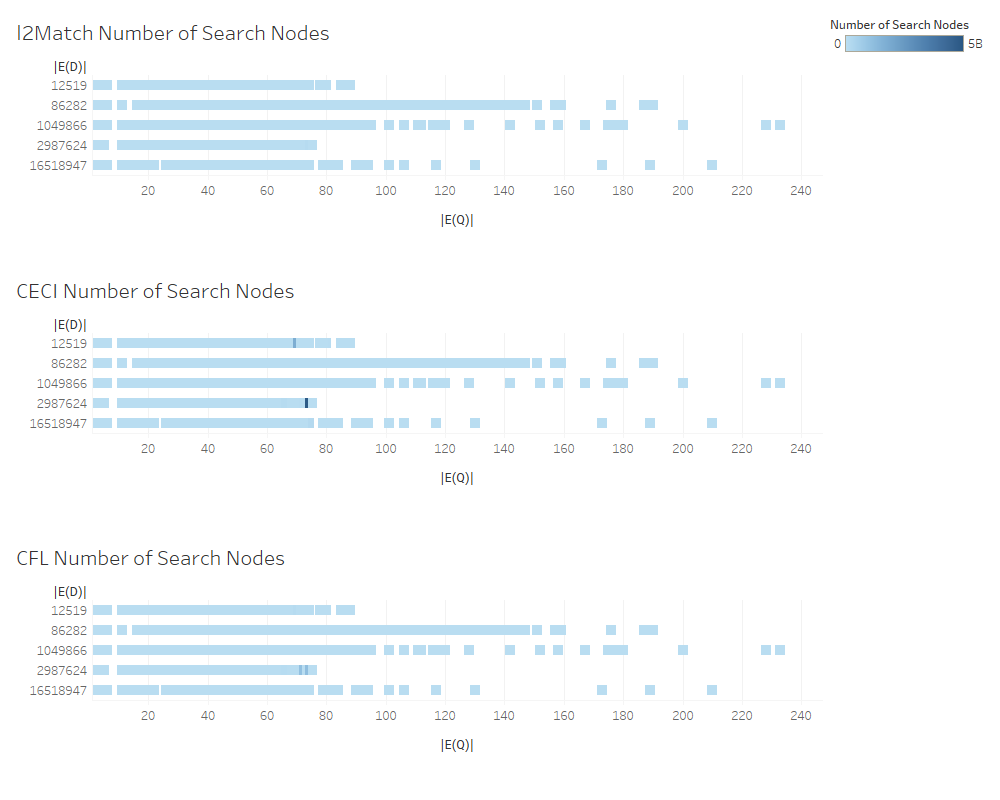}
    \end{subfigure}\quad
    \begin{subfigure}[t]{0.5\textwidth}
        \centering
        \includegraphics[width=\textwidth]{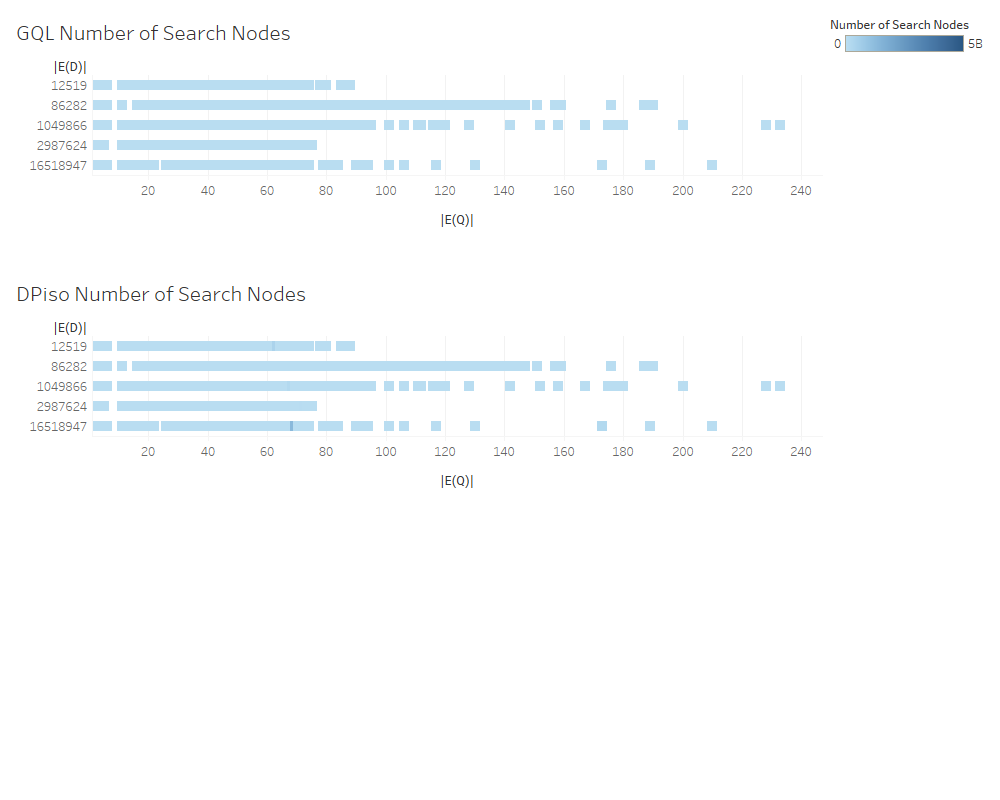}
    \end{subfigure}
    \caption{Number of query edge (x-axis), number of data edges (y-axis), and number of search nodes (color).}
    \label{fig.scatter_callcount}
\end{figure}

Figure \ref{fig.scatter_query} shows that the query time does not increase with the number of query edges. Most of the queries that took 5 seconds or more to solve are concentrated in the similar area for all algorithms. We set the color scale from 0 to 10 seconds, as any longer query time is considered slow in performance. 
In fact, most of the queries are solvable in less than 1 second. The l2Match algorithm demonstrates better performance in terms of query time as it solves the majority of the queries within 0 to 5 seconds. Conversely, GQL has the highest number of query instances solved in 10 seconds and above. 
l2Match and GQL outperform other algorithms in terms of the number of search nodes explored in the enumeration step as evidenced in Figure \ref{fig.scatter_callcount}. 
Furthermore, this proves that the JR optimization method is valid and effective in reducing search branches (the number of search nodes) and query time, as the l2Match algorithm employs an enumeration method that is identical to the CECI algorithm. Further details are discussed in \ref{subsec:JR-result}.

Thus, we classify the hardness of any query graph into four categories:
\begin{enumerate}
   \item Easy query (E), if all algorithm are able to solve it;
   \item Easy or Hard (EH), if more than one algorithms but not all are able to solve the query;
   \item Hard (H), if only exactly one algorithm is able to solve it;
   \item Unsolved (U), if none of the algorithms is able to solve the query within the time limit.
\end{enumerate}

The hardness of the query graphs is not to be confused with the satisfiability, as an unsolved query (U) may be satisfiable or unsatisfiable, but any other classes, such as E, EH, and H, are surely satisfiable.

\begin{table}[t]
    \caption{Satisfiability and Hardness of the Query Graphs.}
    \label{table:hardness}
    \centering
    \resizebox{0.5\textwidth}{!}{
        \begin{tabular}{|l|l r|r r r r|}
            \hline
                \multirow{2}{*}{Dataset} & \multicolumn{2}{c|}{satisfiability} & \multicolumn{4}{c|}{hardness} \\ \cline{2-7}
                & \multicolumn{1}{c}{satisfiable} & \multicolumn{1}{c|}{unsatisfiable} & Easy, E & Easy or Hard, EH & Hard, H & Unsolved, U \\
            \hline \hline
                \textbf{Human} & 900 & 0 & 855 & 42 & 1 & 2 \\
                
                \textbf{Yeast} & 900 & 0 & 892 & 8 & 0 & 0 \\
                
                \textbf{Patents} & 900 & 0 & 757 & 133 & 4 & 6 \\
                
                \textbf{Youtube} & 900 & 0 & 773 & 119 & 7 & 1 \\
                
                \textbf{DBLP} & 900 & 0 & 843 & 44 & 9 & 4 \\
    
                \textbf{ERP} & 164 & 36 & 0 & 0 & 1 & 199 \\
            \hline
        \end{tabular}
    }
\end{table}

\begin{table}[t]
    \caption{Number of Solved Queries.\tnote{1}}
    \label{table:numofsolved}
    \centering
    \resizebox{0.5\textwidth}{!}{
    \begin{threeparttable}
        \begin{tabular}{|l|l|l|l|l|l|l|l|l|l|l|} 
            \hline
                \multirow{3}{*}{Dataset} & \multicolumn{10}{c|}{Number of solved queries} \\ \cline{2-11}
                & \multicolumn{5}{c|}{EH} & \multicolumn{5}{c|}{H} \\ \cline{2-11}
                & l2Match & CECI & CFL & GQL & DPiso & l2Match & CECI & CFL & GQL & DPiso \\
            \hline \hline
                \textbf{Human}  & 27    & 18    & \cellcolor{blue!25}35    & \cellcolor{blue!25}35    & 27    & 0 & 0 & 0 & \cellcolor{blue!25}1 & 0 \\
                
                \textbf{Yeast}  & \cellcolor{blue!25}8 & 4 & 6 & 6 & \cellcolor{blue!25}8   & 0 & 0 & 0 & 0 & 0 \\
                
                \textbf{Patents} & 131   & 130   & \cellcolor{blue!25}132   & 6     & 123   & 0 & 0 & \cellcolor{blue!25}2 & 1 & 1 \\
                
                \textbf{Youtube} & 86    & 51   & 70   & \cellcolor{blue!25}106     & 93    & 0 & 0 & 0 & \cellcolor{blue!25}4 & 3 \\
                
                \textbf{DBLP}   & 34    & 33    & \cellcolor{blue!25}38    & 23    & 32    & 0 & 0 & 1 & \cellcolor{blue!25}7 & 1 \\
    
                \textbf{ERP}    & 0     & 0     & 0     & 0     & 0     & 0 & 0 & 0 & \cellcolor{blue!25}1 & 0 \\
            \hline
        \end{tabular}
        \begin{tablenotes}
            \item[1] This table presents the number of solved queries in terms of hardness of the query: Easy or Hard (EH), and Hard (H). The best results among the competing algorithms in each dataset are highlighted in blue.
        \end{tablenotes}
    \end{threeparttable}
    }
\end{table}

The satifiability and hardness of the query graphs are shown in Table \ref{table:hardness}. We further report the number of solved queries in Table \ref{table:numofsolved}. CFL solves more EH queries in comparison to the other algorithms. Since the enumeration step is the only part of the algorithm that is responsible for verifying the isomorphism of the subgraphs, an effective ordering method allows the enumeration function to prune invalid search branches as early as possible and reduce enumeration time. Clearly, this indicates that the ordering technique of CFL that utilizes the divide-and-conquer method is more effective than the ordering method of the others. Although GQL solves the majority of H queries by employing the DFS ordering method, it is evidenced that the DFS enumeration order slows down the enumeration process in comparison to CFL, which uses a similar enumeration method. Both l2Match and CECI employ the BFS ordering method to sort the query vertices. However, l2Match ranks second in the EH class and fourth in the H class despite being the fastest algorithm. 
In summary, the BFS ordering method is impractical and non-versatile to optimize the enumeration step since it cannot reduce unpromising intermediate results by utilizing the effect of Postponing Cartesian Product~\cite{CFL}.

Research proves that the Constraint Programming (CP) SI algorithm is much more adept at solving computationally challenging queries and estimating the satisfiability of any query \cite{satisfiable}. Nevertheless, results from \cite{hardness} show that the filter-and-verification algorithms are faster than the CP algorithms in terms of easy queries. 
Thus, the essential task is to prove that l2Match is more efficient at solving easy queries. 
Undoubtedly, the experimentation outcomes align with the research objectives and expectations.

\subsection{Average Query Time}
\label{subsec:query-result}

\begin{figure}
    \begin{subfigure}{\columnwidth}
        \centering
        \includegraphics[width=\linewidth]{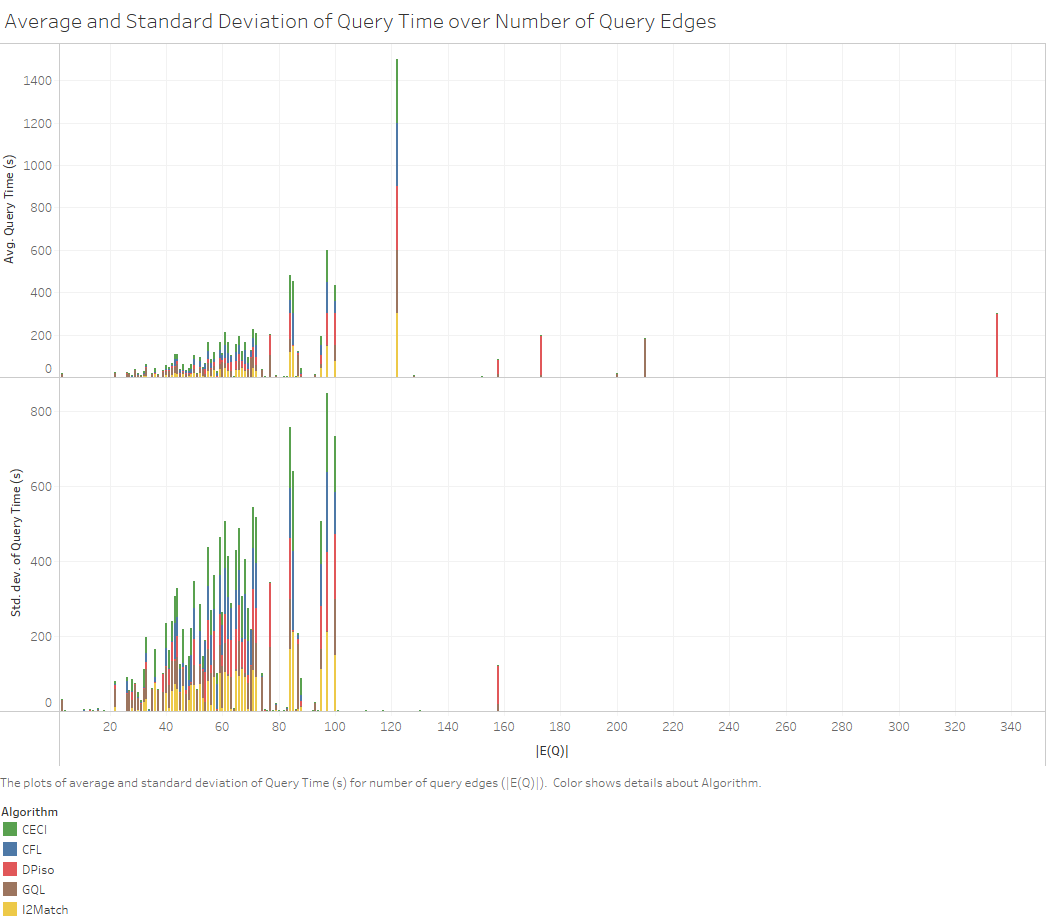}
        \caption{Average query time.}
    \end{subfigure}
    \begin{subfigure}{\columnwidth}
        \centering
        \includegraphics[width=\linewidth]{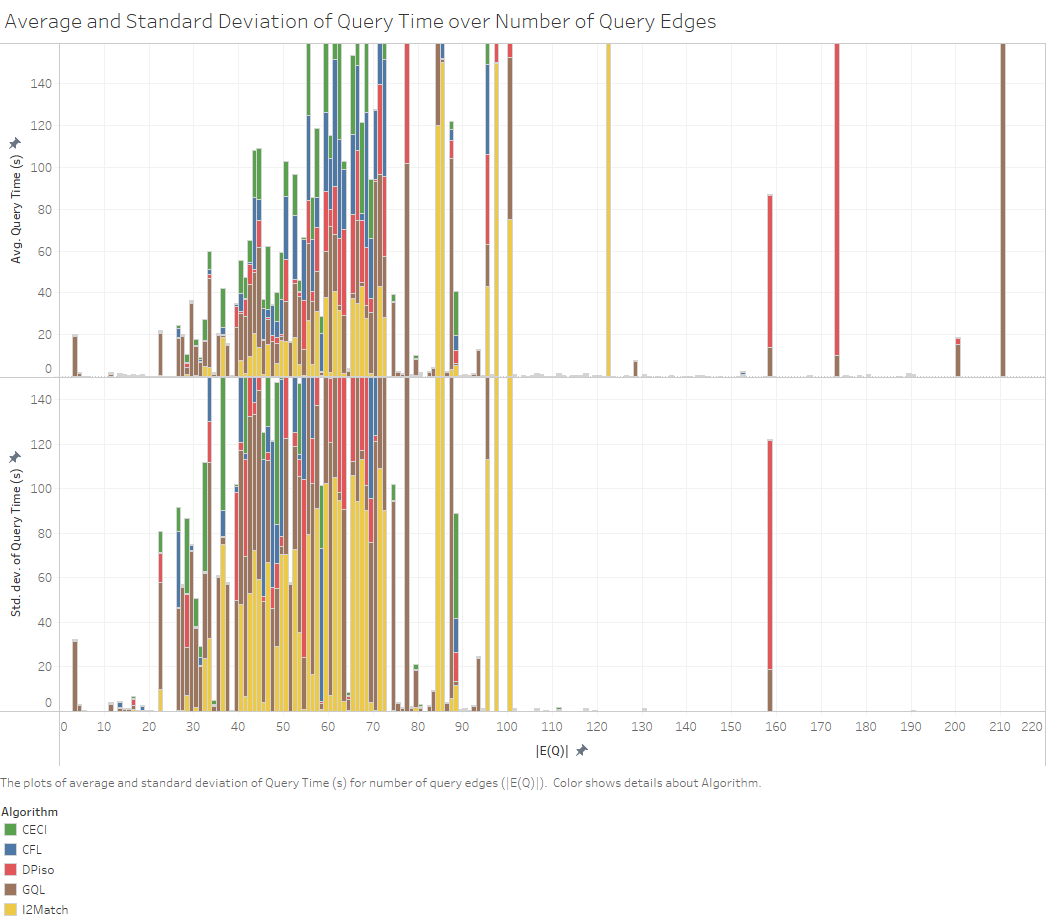}
        \caption{Standard deviation of query time.}
    \end{subfigure}
    \caption{Evaluation on the query time over the number of query edges with l2Match (yellow) as the fastest algorithm followed by CFL (blue), CECI (green), DPiso (red), and lastly GQL (brown).}
    \label{fig.performance-result}
\end{figure}

l2Match significantly outperforms other algorithms, especially on the lower end of the number of query edges, as reported in Figure \ref{fig.performance-result}. 
It achieves performance improvement\footnote{Performance improvement is calculated as the percentage of decrease \(p=100\cdot (t^{algo}-t^{l2Match})\div t^{algo}\) where \(t\) is the average query time and \(algo\) represents each competing algorithm.} over CECI by 19.22\%, CFL by 13.11\%, GQL by 45.21\%, and DPiso by 37.79\%. 
In contrast, the performance of GQL and DPiso are inferior as the number of query edges increases. 
Furthermore, the average and standard deviation of the results suggest that the l2Match and CFL algorithms scale almost linearly to the number of query edges until they reach a maximum value. 
The average degree and density of a query graph increase as the number of edges increases. 
Nonetheless, the number of candidates to verify and the enumeration time decrease as the search is focused on the denser region of the data graph. 
This observation supports the findings that the hardness of a query is peaked in the phase transition between the satisfiable and unsatisfiable areas~\cite{satisfiable}.

\subsection{Effectiveness of LPI, LPF, and BCPRefine Method in Filtering Step}
\label{subsec:l2-CECI}

\begin{figure}[htb]
    \centering
    \includegraphics[width=\columnwidth]{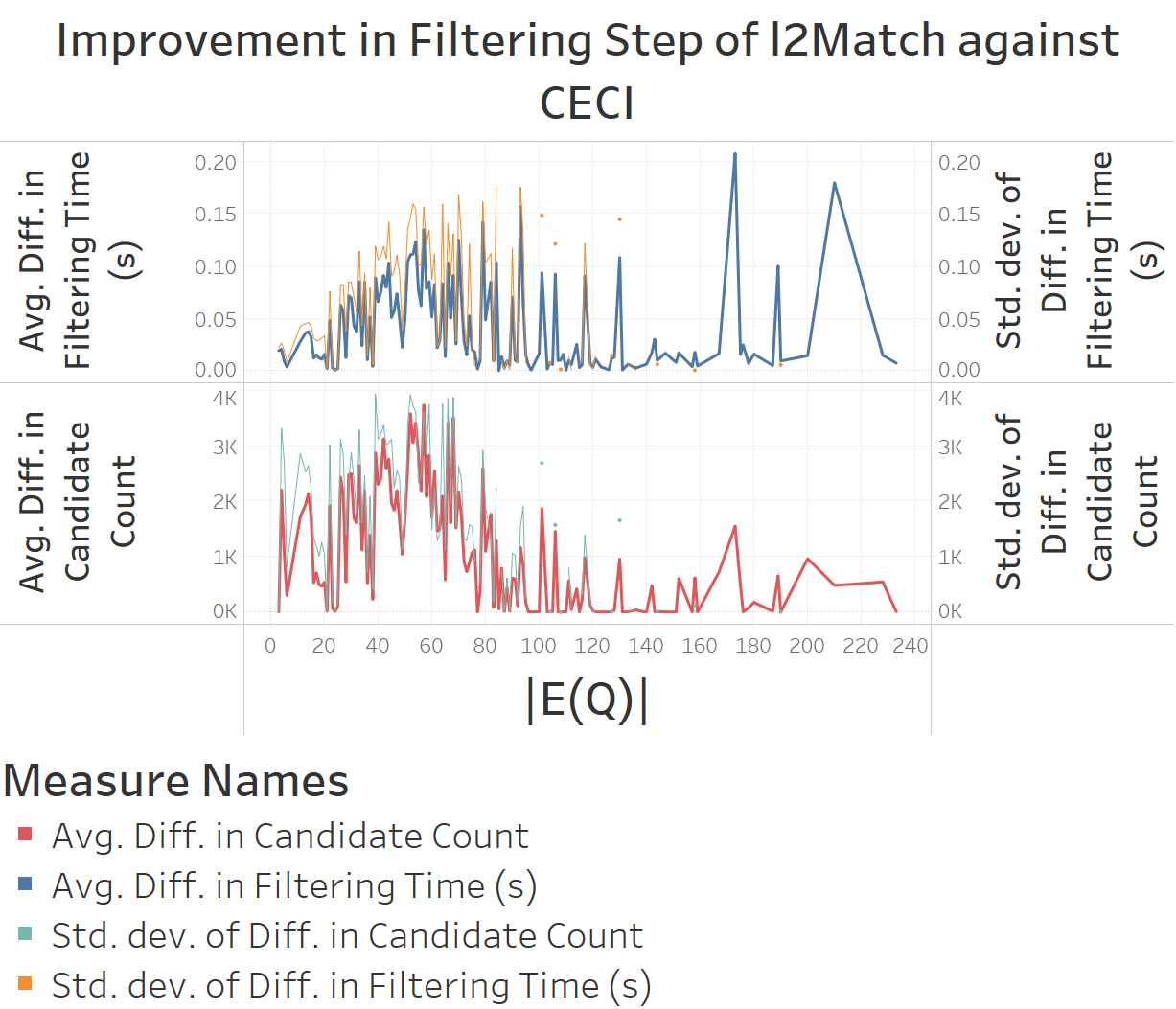}
    \caption{Improvement in filtering step of l2Match against CECI. Top results show the average (blue) and standard deviation (orange) difference in filtering time while bottom results show the average (red) and standard deviation (green) difference in candidate count.}
    \label{fig.l2-CECI}
\end{figure}

Figure \ref{fig.l2-CECI} shows the average performance difference between l2Match and CECI in terms of filtering time and candidate count. A positive difference signifies better performance of l2Match algorithm, and vice versa. 
It is apparent that l2Match outperforms CECI with the implementation of LPI, LPF, and BCPRefine methods that reduce the redundant traversals on the query graph and the expensive removal of invalid candidates. 
On average, l2Match has 14.39\% lesser candidate count and 39.85\% shorter filtering time than CECI\footnote{Average performance different is calculated as the percentage decrease \(\overline{p}=100(\Sigma_{i=1}^{N} (x_i^{CECI}-x_i^{l2Match})\div x_i^{CECI})\div N\) where \(x\) is the value of candidate count or filtering time, and \(N\) is the number of easy (E) queries.}.

\subsection{Improvement using LPI and LPF Method in Filtering Step}
\label{subsec:LPF-result}

\begin{figure}
    \begin{subfigure}[t]{\columnwidth}
        \centering
        \includegraphics[width=\textwidth]{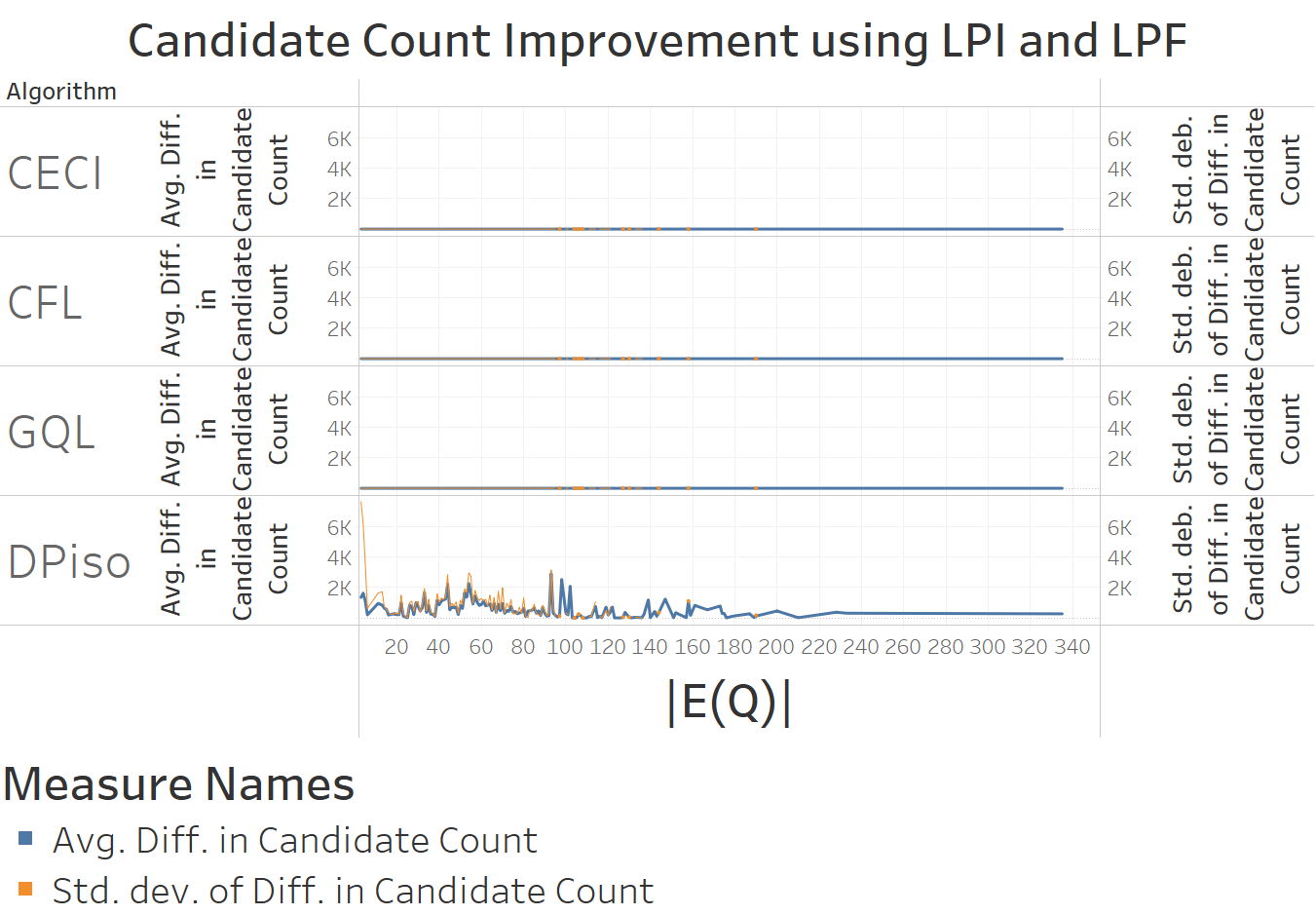}
        \caption{Candidate count.}
    \end{subfigure}\quad
    \begin{subfigure}[t]{\columnwidth}
        \centering
        \includegraphics[width=\textwidth]{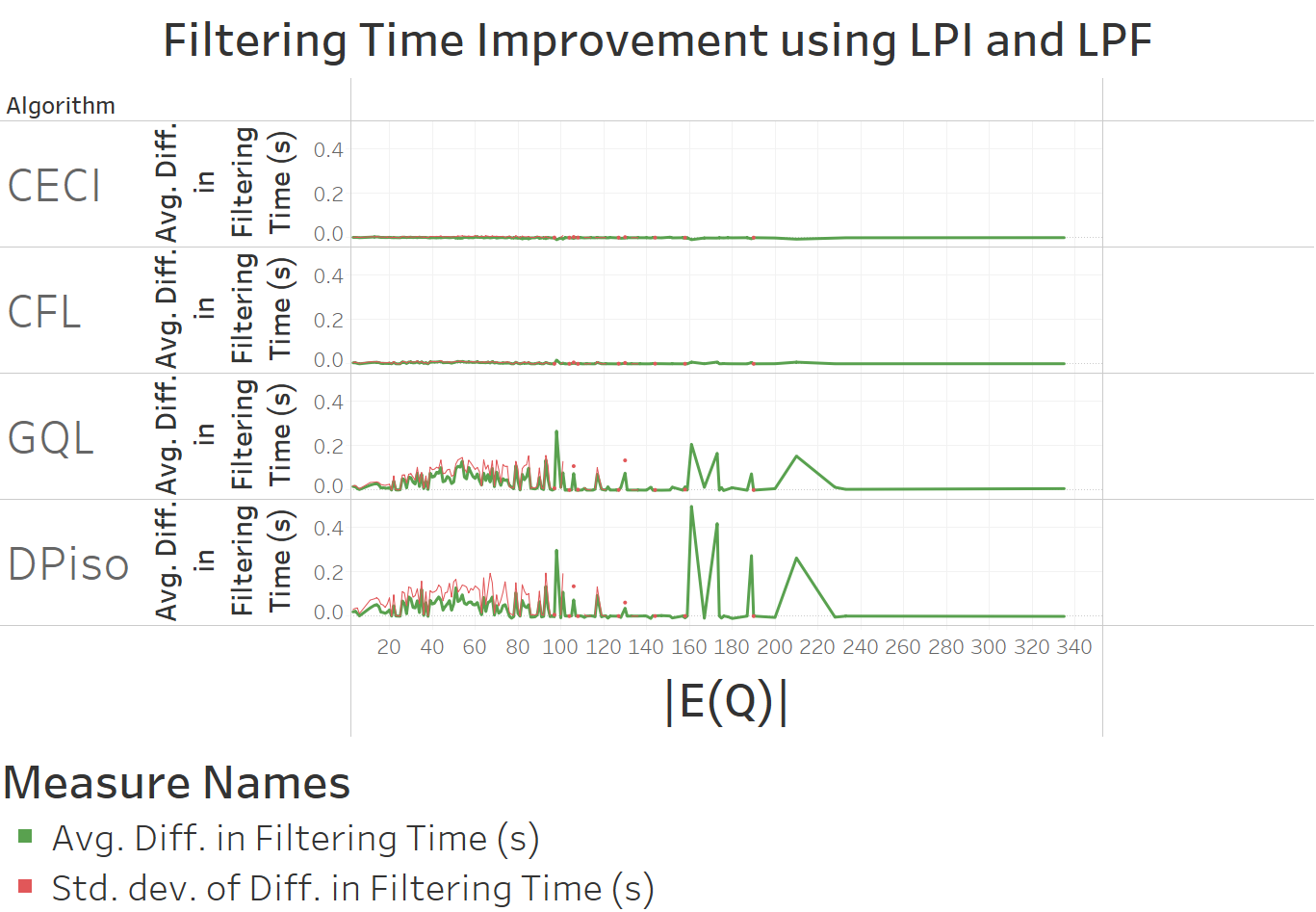}
        \caption{Filtering time.}
    \end{subfigure}
    \caption{Evaluation on the average (green, thick line) and standard deviation (red, thin line) of difference in the candidate count and the filtering time in the filtering step over the number of query edges for CECI, CFL, GQL, and DPiso algorithms (from top to bottom) optimized with LPI and LPF methods.}
    \label{fig.compareLPF}
\end{figure}

The effectiveness of the LPI and LPF methods are evaluated and portrayed in Figure~\ref{fig.compareLPF}. 
One positive trend (DPiso) is seen in the candidate count, and two strong positive trends (GQL and DPiso) are seen in the filtering time. 
The candidate count of the DPiso algorithm is decreased by 16.19\%, on averagely as a result of the optimization in filtering using the LPI and LPF methods\footnote{Average performance is calculated directly as the percentage decrease \(\overline{p}=100(\Sigma_{i=1}^{N} (x_i^{ORI}-x_i^{LPF})\div x_i^{ORI})\div N\) where \(x\) is the value of the query time or the number of search nodes, \(ORI\) and \(LPF\) represents the original and optimized algorithm respectively, and \(N\) is the number of easy (E) queries.}; 
however, it shows no changes in CECI, CFL, and GQL algorithms. It indicates that the pruning power of the LPF and NLF method are equivalent. Besides, LDF's pruning power is weaker than that of the LPF and NLF methods, as observed in the positive difference between DPiso-LPF and DP-iso.

LPI and LPF methods shorten the filtering time by 9.60\% (CFL), 20.75\% (GQL), and 16.94\% (DPiso), respectively, but increase the filtering time of CECI by 0.05\%. We speculate that the overhead of accessing a query vertex's neighbors stored in the LPI outweighs the time taken to scan every neighbor of a query vertex. This is true when the query vertex has a minimal number of neighbors. In spite of the shortcoming, LPI and LPF improve the filtering step by cutting down on the number of neighbors to access and verify, ultimately decreasing the filtering time.

\subsection{Improvement using JR Method in Enumeration Step (RO4)}
\label{subsec:JR-result}

\begin{figure}
    \begin{subfigure}[t]{\columnwidth}
        \centering
        \includegraphics[width=\linewidth]{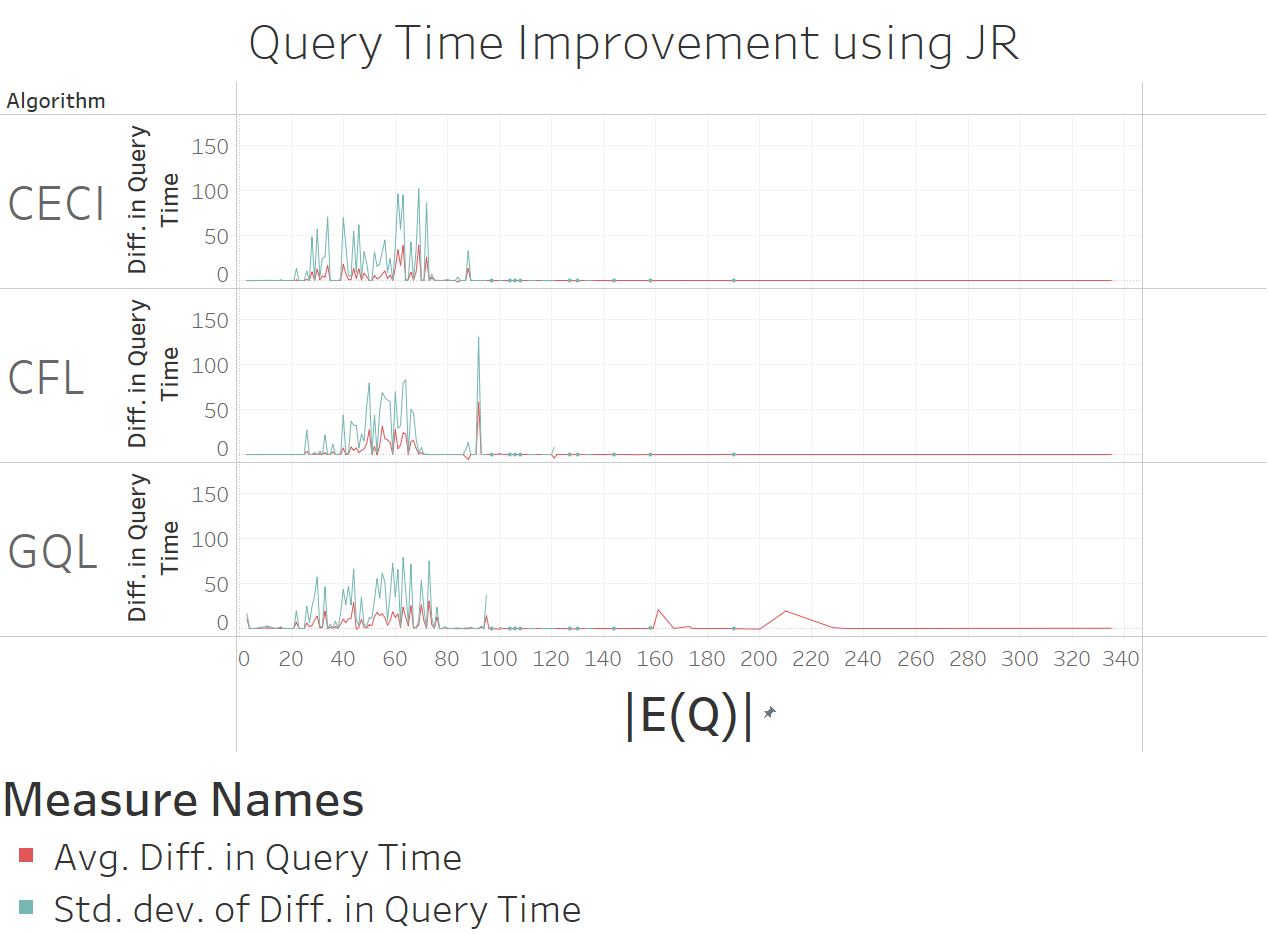}
        \caption{Query time.}
    \end{subfigure}\quad
    \begin{subfigure}[t]{\columnwidth}
        \centering
        \includegraphics[width=\linewidth]{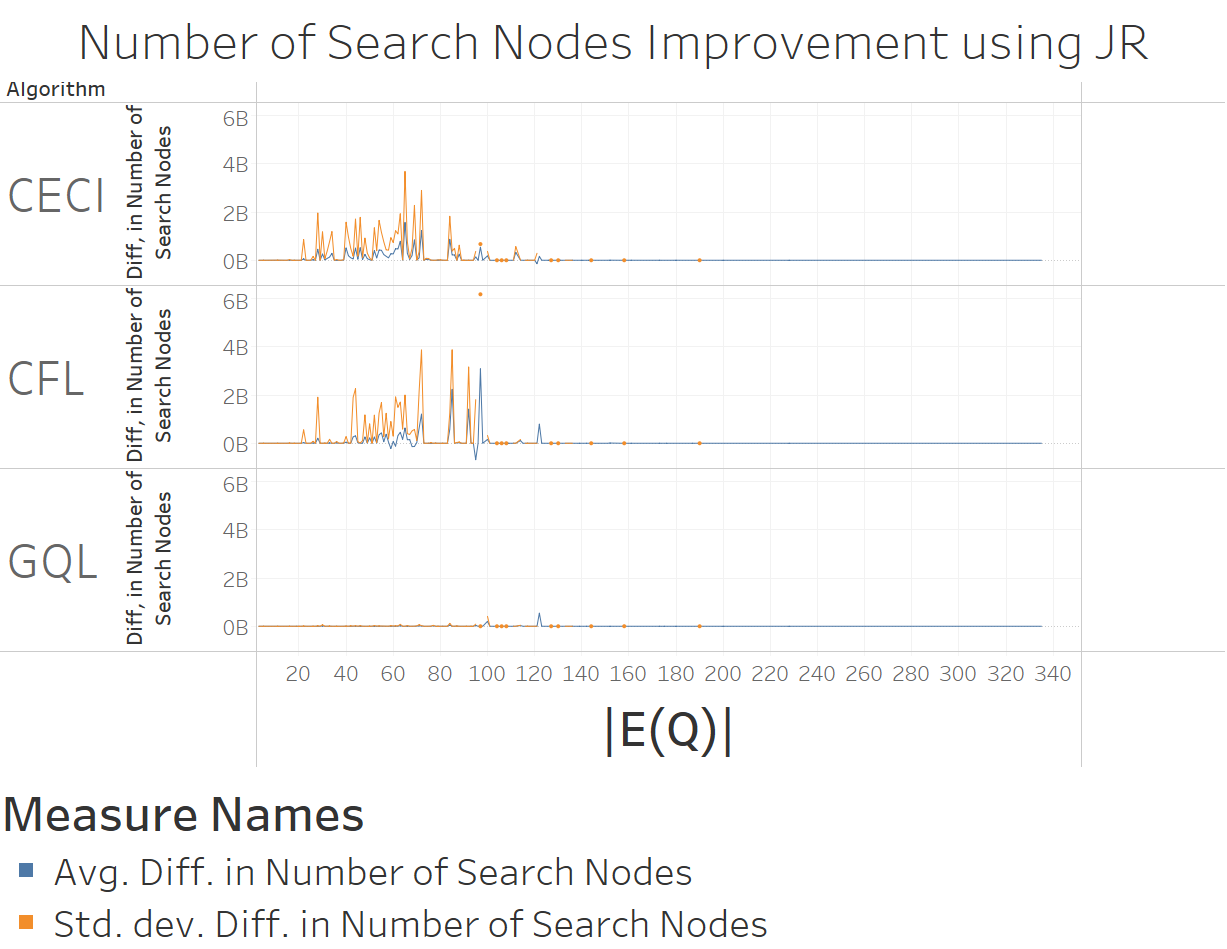}
        \caption{Number of search nodes.}
    \end{subfigure}
    \caption{Evaluation on the average (blue) and standard deviation (orange) of difference in the query time and the number of search nodes explored in the enumeration step over the number of query edges for CECI, CFL, and GQL algorithms (from top to bottom) optimized with JR method.}
    \label{fig.compareJR}
\end{figure}

The improvements in the enumeration step of CECI, CFL, and GQL algorithms achieved using the JR method are visualized in Figure~\ref{fig.compareJR}. 
DP-iso is not compatible with the JR optimization method, as JR only works on static enumeration order but not dynamic enumeration order. 
JR demonstrates positive improvements when paired with CECI algorithm. 
On the other hand, CFL-JR occasionally performs worse than the original in terms of query time when the number of query edges range between 80 and 90. 
Similarly, the query time of the original algorithm peaked in the same range in Figure~\ref{fig.performance-result} of Section~\ref{subsec:query-result}. 
Nonetheless, the performance of CFL-JR is mostly better than CFL. 
The difference in the number of search nodes is insignificant for GQL algorithm even though the query time of GQL-JR is, on average, shorter than GQL. 
The ordering method of GQL prioritizes neighboring vertices with minimum number of candidates. 
Consequently, the closest backward neighbor of each query vertex in the enumeration order is placed right before itself (distance of 1). 
Further inspections reveal that the JR method achieves better optimization in reducing exploration of redundant search branches when the spread and distance between the source and target query vertices of a jump is further.

The average improvement achieved with JR method in the enumeration step are 12.35\% (CECI), 5.88\% (CFL) and 41.62\% (GQL) in terms of the query time, and 46.47\% (CECI), 55.26\% (CFL) and 52.94\% (GQL) in terms of the number of search nodes\footnote{Average performance different calculated as the percentage decrease \(\overline{p}=100(\Sigma_{i=1}^{N} (x_i^{ORI}-x_i^{JR})\div x_i^{ORI})\div N\) where \(x\) is the value of the query time or the number of search nodes, \(ORI\) and \(JR\) represents the original and optimized algorithm respectively, and \(N\) is the number of easy (E) queries.}.

\section*{Acknowledgment}

This research was supported by Monash University Malaysia Merit Scholarship Scholarship. I cannot express enough thanks to my supervisors Assoc Prof Wong Kok Sheik and Dr Soon Lay Ki. Their tremendous guidance and advice are the keys to accomplish this research. A special thank to Julie Holden, who continuously assist me to structure and ameliorate my thesis. I would also like to acknowledge with much appreciation to the advice given by the Milestone Review Panel of Monash University. They inspire me to explore further in Subgraph Matching field and ensure that I am not swayed from my scope.

\printbibliography[title=References]

\end{document}